\documentclass[conference]{IEEEtran}
\usepackage[utf8]{inputenc}
\usepackage{url}
\usepackage{bm}
\usepackage[margin=1in]{geometry} 
\usepackage{color}    
\usepackage{latexsym, amssymb, amscd, amsfonts, amsthm, amsmath, graphicx}

\usepackage{epsfig}
\usepackage{color}                 
\usepackage{epsfig}
\usepackage{mathrsfs}
\usepackage{mathtools}
\usepackage{caption}
\usepackage{subcaption}
\usepackage{graphicx}
\usepackage{wrapfig}
\usepackage{pst-node}
\usepackage{tikz-cd} 
\usepackage{xcolor}
\usepackage{needspace}
\usepackage[numbers]{natbib}

\usepackage[ruled]{algorithm2e}
\theoremstyle{plain}
\newtheorem{theorem}{Theorem}[section]
\newtheorem{lemma}[theorem]{Lemma}
\newtheorem{definition}[theorem]{Definition}

\newcommand{\R}{\mathbb{R}}      

\newcommand{\cF}{\mathscr{F}}    

\newcommand{\p}{\mathbb{P}}

\newcommand{\JP}{\text{J}_\mathcal{P}}
\newcommand{\JW}{\text{J}_\mathcal{W}}

\newcommand{\sP}{\mathcal{P}}
\newcommand{\sW}{\mathcal{W}}

\usepackage[utf8]{inputenc}

\newcommand{\argmin}{\operatorname*{arg\,min}}
\newcommand{\argmax}{\operatorname*{arg\,max}}

\newif\ifanon

\anonfalse

\begin{document}
\onecolumn
\par\noindent
\thispagestyle{empty}

\noindent \mbox{\copyright{}} 2018 IEEE. Personal use of this material is permitted. Permission from IEEE must be obtained for all other uses, in any current or future media, including reprinting/republishing this material for advertising or promotional purposes, creating new collective works, for resale or redistribution to servers or lists, or reuse of any copyrighted component of this work in other works.
\bigskip
\par\noindent To appear in the Proceedings of ICDMW 2018.

\bigskip

\par\noindent
\bigskip
\par\noindent
\newpage

\twocolumn
\bstctlcite{IEEEexample:BSTcontrol}
\title{Maximally Consistent Sampling and the\\Jaccard Index of Probability Distributions}
\ifanon
\author{Anonymous\\Anonymous}
\else
\author{
\IEEEauthorblockN{Ryan Moulton}
\IEEEauthorblockA{Google Inc.\\
ryanmoulton@gmail.com}
\and
\IEEEauthorblockN{Yunjiang Jiang}
\IEEEauthorblockA{JD.COM, Silicon Valley Research Center\\
yunjiangster@gmail.com}
}
\maketitle
\thispagestyle{plain}
\pagestyle{plain}
\begin{abstract} We introduce simple, efficient algorithms for computing a MinHash of a probability distribution, suitable for both sparse and dense data, with equivalent running times to the state of the art for both cases. The collision probability of these algorithms is a new measure of the similarity of positive vectors which we investigate in detail. We describe the sense in which this collision probability is optimal for any Locality Sensitive Hash based on sampling. We argue that this similarity measure is more useful for probability distributions than the similarity pursued by other algorithms for weighted MinHash, and is the natural generalization of the Jaccard index.
\end{abstract}
\section{Introduction}
MinHashing\cite{classicminhash} is a popular Locality Sensitive Hashing algorithm for clustering and retrieval on large datasets. Its extreme simplicity and efficiency, as well as its natural pairing with MapReduce and key-value datastores, have made it a basic building block in many domains, particularly document clustering\cite{classicminhash}\cite{broder1998filtering} and graph clustering\cite{graphclustering}\cite{pregelgraphclustering}.

Given a finite set $U$ of size $n$, and a uniformly random bijection $\pi: U \to \{1,2, \ldots, n\}$, the stochastic map $X \mapsto \argmin_{x \in X} \pi(x)$ enjoys the following stability property with respect to $X$:
\begin{align*}
\Pr\left[\argmin_{x\in X} \pi(x) = \argmin_{y \in Y} \pi(y)\right] = \text{J}(X,Y).
\end{align*}
Here $X$, $Y$ are both subsets of $U$, and $\text{J}(X, Y)$ is the well-known Jaccard similarity \cite{jaccardindex} defined by
\begin{align*}
\text{J}(X,Y) = \frac{|X \cap Y|}{|X \cup Y|}.
\end{align*}

Practically, this random permutation is generated by applying some hash function to each $i \in U$ with a fixed random seed, hence ``MinHashing".

In order to hash objects other than sets, Chum et al.\cite{expweight} introduced two algorithms for incorporating weights in the computation of MinHashes. The first algorithm associates constant global weights with the set members, suitable for idf weighting. The collision probability that results is $\frac{\sum_{i \in X \cap Y} w_i}{\sum_{i \in X \cup Y} w_i}$. The second algorithm computes MinHashes of vectors of positive integers, yielding a collision probability of
\begin{align*}
\JW(x,y) &= \frac{\sum_i\min\left(x_i, y_i\right)}{\sum_i\max\left(x_i, y_i\right)} 
\end{align*}
Subsequent works\cite{wminhash}\cite{denseminhash}\cite{expected-constant} have improved the efficiency of the second algorithm and extended it to arbitrary positive weights, while still achieving $\JW$ as the collision probability.

$\JW$ is one of several generalizations of the Jaccard index to non-negative vectors. It is useful because it is monotonic with respect to the $\text{L}_1$ distance between $x$ and $y$ when they are both $\text{L}_1$ normalized, but it is unnatural in many ways for probability distributions.

If we convert sets to binary vectors $x, y$, with $x_i, y_i\!\in\!\{0,1\}$, then $\JW(x, y) = \text{J}(X,Y)$. But if we convert these vectors to probability distributions by normalizing them so that $x_i\!\in\!\left\{0,\frac{1}{|X|}\right\},\  y_i\!\in\!\left\{0, \frac{1}{|Y|}\right\}$, then $\JW(x,y)\neq\text{J}(X,Y)$ when $|X|\neq |Y|$. The correspondence breaks and it no longer generalizes the Jaccard index. Instead, for $|X| > |Y|$, 
\begin{align}
\label{jw-less-than-j}
\JW(x,y) = \frac{|X\cap Y|}{|X\setminus Y| + |X|} < \text{J}(X,Y).
\end{align}
As a consequence, switching a system from an unweighted MinHash to a MinHash based on $\JW$ will generally decrease the collision probabilities.

Furthermore, $\JW$ is insensitive to important differences between probability distributions.  It counts all differences on an element in the same linear scale regardless of the mass the distributions share on that element. For instance, $\JW\left((a,b,c,0), (a,b,0, c)\right)=\JW\left((a+c,b), (a,b+c)\right)$. This makes it a poor choice when the ultimate goal is to measure similarity using an expression based in information-theory or likelihood where having differing support typically results in the worst possible score.

For a drop-in replacement for the Jaccard index that treats its input as a probability distribution, we'd like it to have the following properties.
\begin{enumerate}
\item Scale invariant in both arguments.
\item Not lower than the Jaccard Index when applied to discrete uniform distributions.
\item Sensitive to changes in support, in a similar way to information-based measures.
\item Easily achievable as a collision probability.
\end{enumerate}
$\JW$ fails all but the last.
\begin{enumerate}
\item It isn't scale invariant, $\JW(\alpha x, y) \neq \JW(x, y)$.
\item If the vectors are normalized to make it scale invariant, the values drop below the corresponding Jaccard index (equation~\ref{jw-less-than-j}.)
\item It is insensitive to changes in support.
\item Good algorithms exist, but they are non-trivial.
\end{enumerate}

\begin{algorithm}[t]
\DontPrintSemicolon
\SetKw{KwWhere}{where}
\SetKwInOut{Input}{input}\SetKwInOut{Output}{output}
\Input{Vector $\mathbf{x}$, Seed $\mathbf{s}$}
\Output{Stable sample from $\mathbf{x}$}
\ForEach{i \KwWhere$x_i>0$}{
$e_i \leftarrow \frac{- \log\left(\text{UniformNonZeroFloat}(i, s)\right)}{x_i}$\;
}
\Return{$\argmin_i e_i$}
\caption{$\sP$-MinHash. We require only a single exponentially distributed hash per nonzero element.}
\label{p-minhash-algorithm}
\end{algorithm}

Existing work has thoroughly explored improvements to Chum et al.'s second algorithm, while leaving their first untouched. In this work we instead take their first algorithm as a starting point. We extend it to arbitrary positive vectors (rather than sets with constant global weights) and analyze the result. In doing so, we find that the collision probability is a new generalization of the Jaccard Index to positive vectors, which we here call $\JP$.
\begin{align} \label{JP_double_sum_formula}
\JP(x,y) &= \sum_{\{ i\ :\ x_i,y_i>0 \}} \frac{1}{\sum_{j} \max\left(\frac{x_j}{x_i}, \frac{y_j}{y_i}\right)}.
\end{align}
The names used here, $\JW$ and $\JP$, are chosen to reflect how each function interprets $x$ and $y$, and the conditions under which they match the original Jaccard index. $\JW$ treats a difference in magnitude the same as any other difference, so treats vectors as ``weighted sets."  $\JP$ is scale invariant, so any input is treated the same as a corresponding probability distribution.

The primary contribution of this work is to derive and analyze $\JP$, and to show that in  many situations where the objects being hashed are probability distributions, $\JP$ is a more useful collision probability than $\JW$.

We will describe the sense in which $\JP$ is an optimal collision probability for any LSH based on sampling. We will prove that if the collision probability of a sampling algorithm exceeds $\JP$ on one pair, it must sacrifice collisions on a pair that has higher $\JP$.

We will motivate $\JP$'s utility by showing experimentally that it has a tighter relationship to the Jensen-Shannon divergence than $\JW$, and is more closely centered around the Jaccard index than $\JW$. We will even show empirically that in some circumstances, it is better for retrieving documents that are similar under $\JW$ than $\JW$ itself (and consequently, sometimes better for retrieving based on $\text{L}_1$-distance.)

\section{$\sP$-MinHash and its Match Probability}

Let $h: [n] \to (0,1]$ be a pseudo-random hash mapping every element $1 \le i \le n$ to an independent uniform random value in $(0,1]$. Over a non-negative vector $x$, define
\begin{align*}
H(x) \coloneqq \argmin_{i} \frac{-\log h(i)}{x_i}
\end{align*}
For brevity, we will let $1/0 \coloneqq \infty$. Each term is an exponentially distributed random variable with rate $x_i$,
\begin{align*}
\Pr\left[\frac{-\log h(i)}{x_i} < y \right] = 1 - e^{-x_i y},
\end{align*}
so it follows that
\begin{align*}
\Pr\left[H(x) = i\right] = \frac{x_i}{\sum_i x_i}.
\end{align*}
This well known and beautiful property of exponential random variables derives from the fact that $\Pr\left[\min(x,y) > \alpha\right] = \Pr\left[x > \alpha\right]\Pr\left[y > \alpha\right]$.

\begin{figure*}
    \centering
    \begin{subfigure}[b]{0.23\textwidth}
        \centering
        \includegraphics[width=0.9\textwidth]{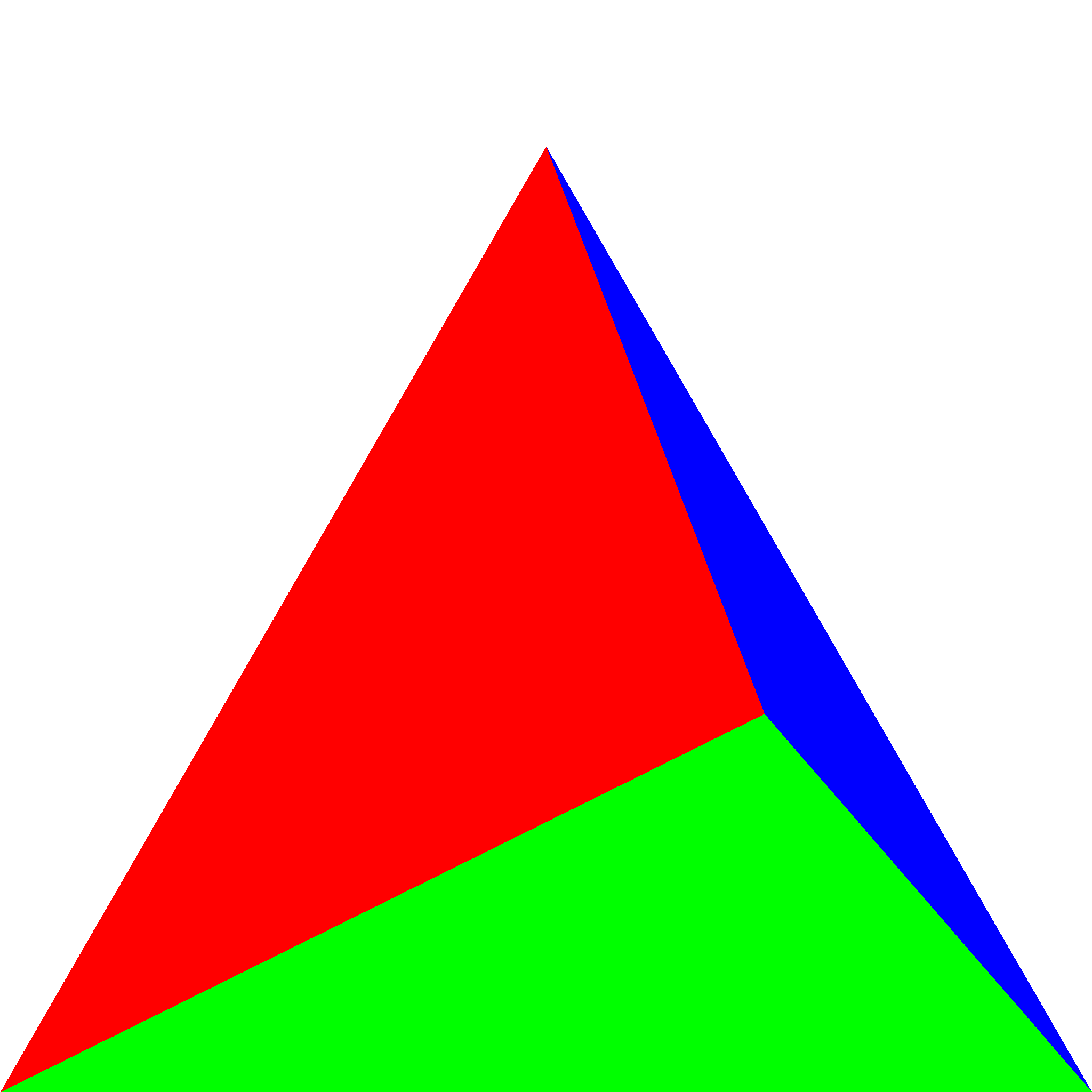}
        \caption{$\Pr\left[ H(x) = i\right]$}
    \end{subfigure}\hfill
    \begin{subfigure}[b]{0.23\textwidth}
        \centering
        \includegraphics[width=0.9\textwidth]{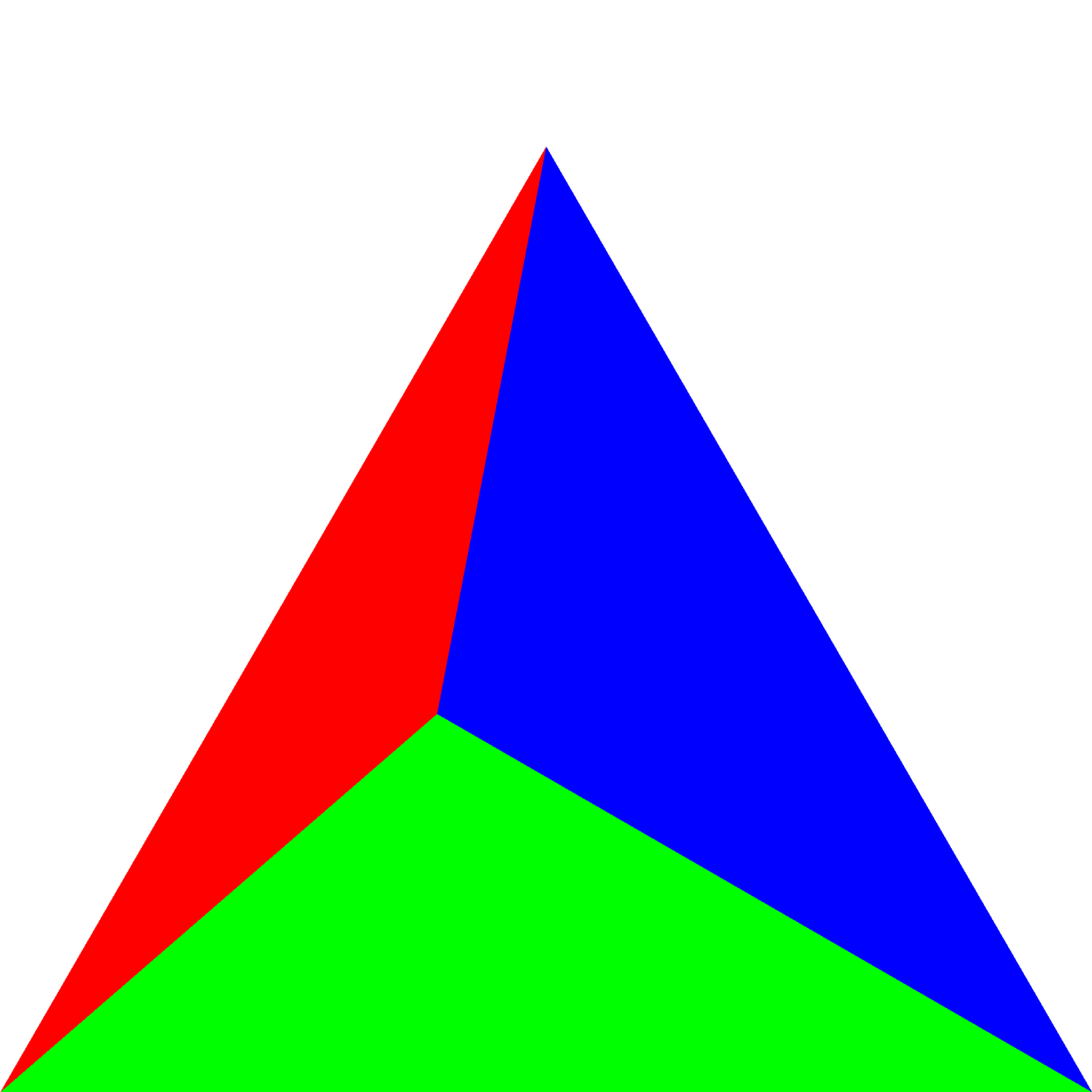}
        \caption{$\Pr\left[ H(y) = i\right]$}
    \end{subfigure}\hfill
    \begin{subfigure}[b]{0.23\textwidth}
        \centering
        \includegraphics[width=0.9\textwidth]{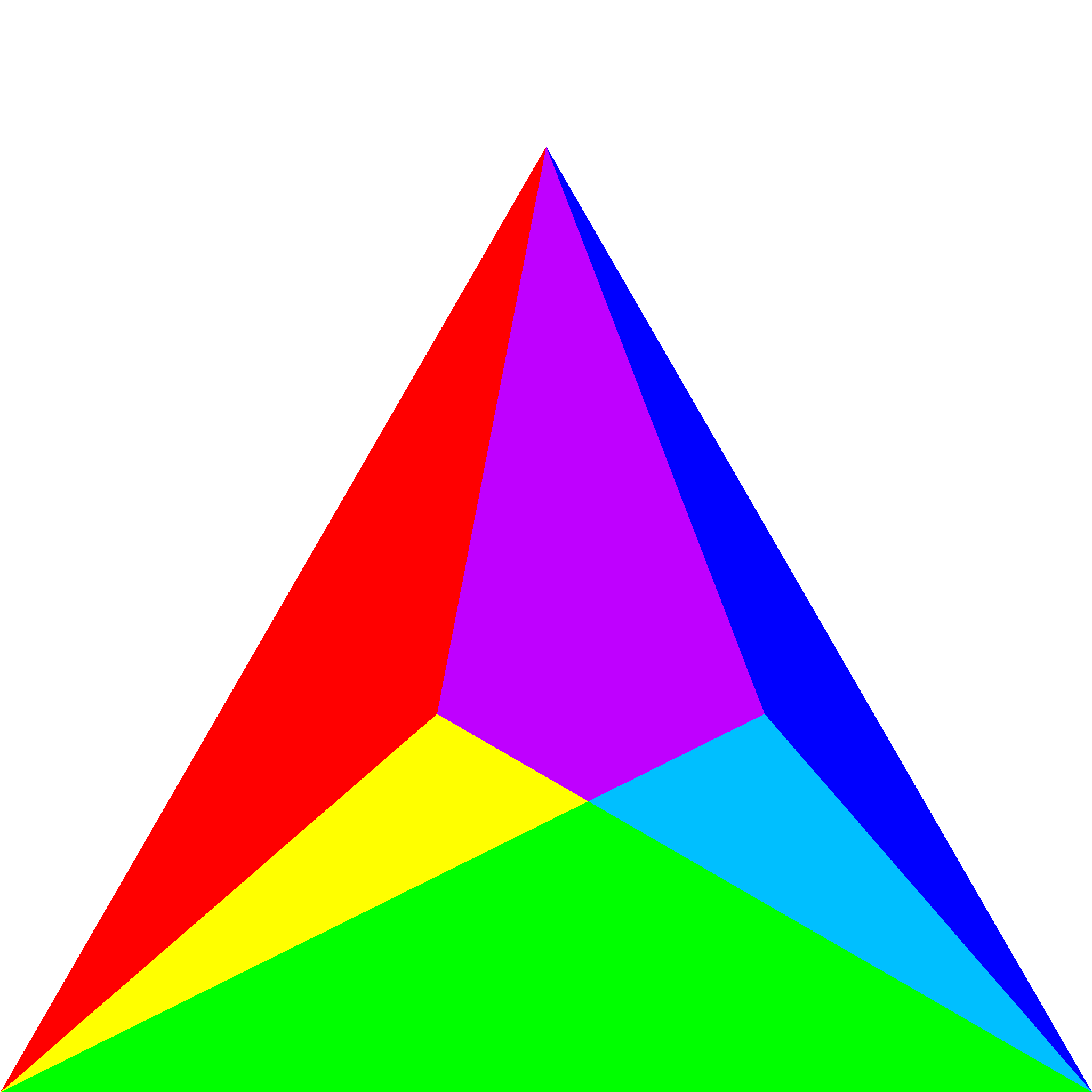}
        \caption{$\Pr\left[ H(x),H(y) = i,j\right]$}
    \end{subfigure}\hfill
    \begin{subfigure}[b]{0.23\textwidth}
        \centering
        \label{jp-simplex-area}
        \includegraphics[width=0.9\textwidth]{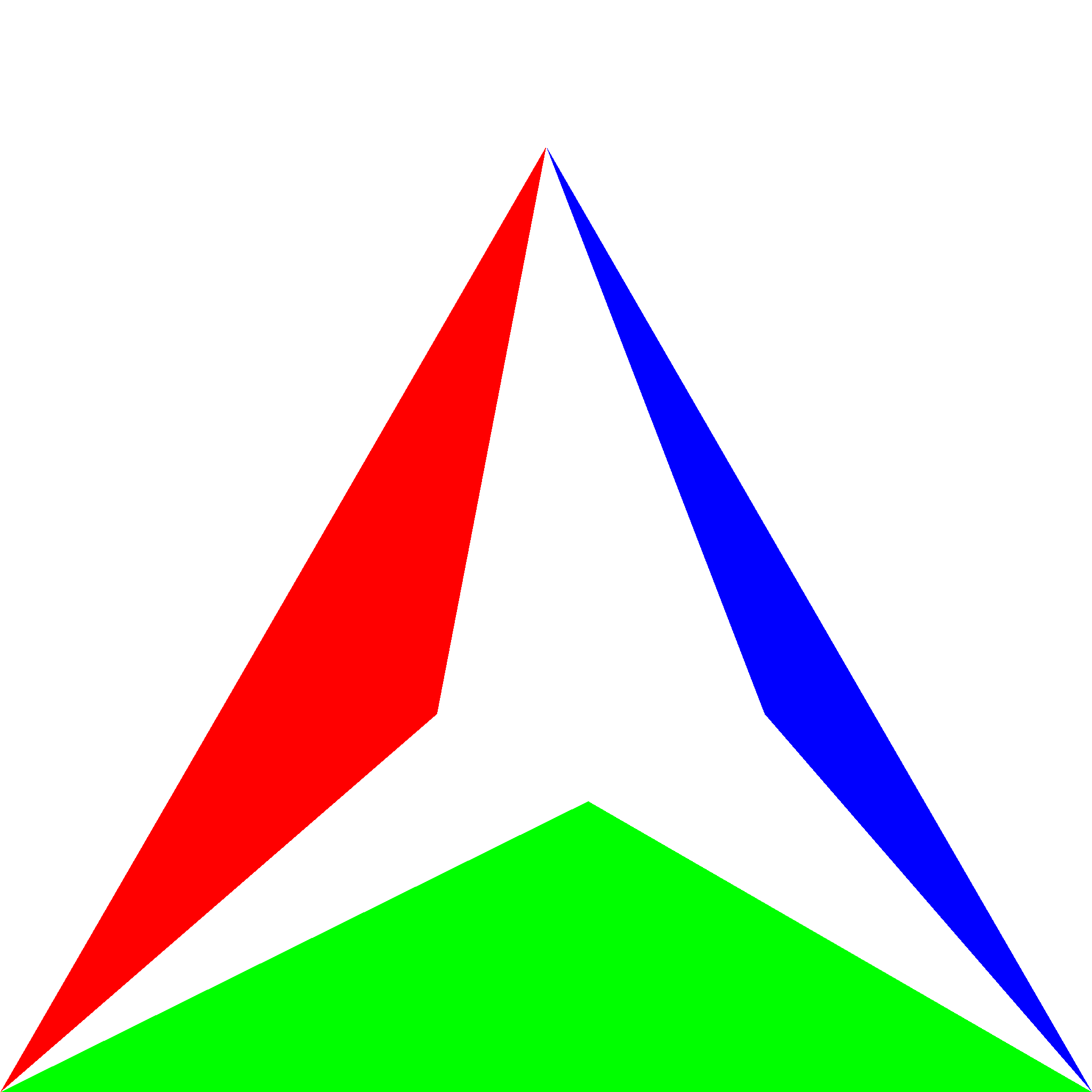}
        \caption{$\Pr\left[ H(x) = H(y) = i\right]$}
    \end{subfigure}
    \caption{The $\sP$-MinHash process can be interpreted geometrically as dividing the unit simplex into smaller simplexes proportional to the mass of each term, then selecting the same point in each simplex as its representative. Pictured here are $x=(0.5,0.4,0.1), y=(0.2,0.4,0.4)$. Colors are assigned to unique values of $i$, or $(i,j)$ in the case of the joint distributions. The sum of the remaining colored areas in (d) is proportional to $\JP(x,y)$.}
\label{simplexfigure}
\end{figure*}

\begin{theorem} \label{Collision_Probability}
For any two nonnegative vectors $x, y \in \R_+^n$,
$$\Pr\left[H(x) = H(y)\right] = \JP(x,y).$$
\end{theorem}
\begin{proof}
Any monotonic transform on $\frac{-\log h(i)}{x_i}$ will not change the arg\,min, so multiplying each $x_i$ by a positive $\alpha$ won't either. $H(\alpha x) = H(x)$. Thus for $x_i,y_i>0$
\begin{align*}
\Pr\left[i = H(x) = H(y)\right] &= \Pr\left[i = H\left(\frac{x}{x_i}\right) = H\left(\frac{y}{y_i}\right)\right]
\end{align*}
By definition, $H\left(\frac{x}{x_i}\right) = H\left(\frac{y}{y_i}\right) = i$ means $\frac{-\log h(i)}{x_i / x_i} \leq \frac{-\log h(j)}{x_j/x_i}$ and $\frac{-\log h(i)}{y_i / y_i} \leq \frac{-\log h(j)}{y_j/y_i}$, for $j \neq i$, or equivalently,
\begin{align*}
-\log h(i) &\leq \min\left(\frac{-\log h(j)}{x_j/x_i}, \frac{-\log h(j)}{y_j/y_i}\right)\\
&\leq \frac{-\log h(j)}{\max\left(\frac{x_j}{x_i},\frac{y_j}{y_i}\right)}.
\end{align*}
Now we desire a new vector $z^i$ such that $H(z^i) = i$ if and only if $H(x/x_i) = H(y/y_i) = i$. This requires that $\frac{-\log h(j)}{z^i_j/z^i_i} = \frac{-\log h(j)}{\max\left(\frac{x_j}{x_i},\frac{y_j}{y_i}\right)}$. Thus for fixed $i$, $z^i_j = \frac{\max\left(\frac{x_j}{x_i},\frac{y_j}{y_i}\right)}{\sum_k \max\left(\frac{x_k}{x_i},\frac{y_k}{y_i}\right)}$. Consequently,
\begin{align*}
\Pr\left[H(z^i) = i\right] = \frac{1}{\sum_j\max\left(\frac{x_j}{x_i},\frac{y_j}{y_i}\right)}
\end{align*}
Repeating this process for all $i$ in the intersection yields
\begin{align*}
\Pr\left[H(x) = H(y)\right] = \sum_{i} \frac{1}{\sum_{j} \max\left(\frac{x_j}{x_i}, \frac{y_j}{y_i}\right)}\quad\qedhere
\end{align*}\end{proof}

Continuing the notational convention, we will refer to hashing algorithms that achieve $\JP$ as their pair collision probability as $\sP$-MinHashes, and algorithms that achieve $\JW$ as $\sW$-MinHashes.

\section{Intuition about $\JP$}

While $\JP$'s expression is superficially awkward, we can aid intuition by representing it in other ways. The simplest interpretation is to view it as a variant of $\JW$. We rewrite $\JW$ allowing one input to be rescaled before computing each term,
\begin{align*}
\JW(x,y,\alpha) = \sum_i \frac{\min(\alpha_i x_i, y_i)}{\sum_j\max(\alpha_i x_j, y_j)}
\end{align*}
and choose the vector $\alpha$ to maximize this generalized $\JW$. If $\alpha_i x_i > y_i$, increasing $\alpha_i$ raises only the denominator. If $\alpha_i x_i < y_i$, increasing $\alpha_i$ raises the numerator proportionally more than the denominator. So the optimal $\alpha$ sets $\alpha_i x_i = y_i$, and results in $\JP$.

We can derive a more powerful representation by viewing the $\sP$-MinHash algorithm itself geometrically. A vector of $k+1$ exponential random variables, when normalized to sum to 1, is a uniformly distributed random point on the unit $k$-simplex. Every point in the unit $k$-simplex is also a probability distribution over $k+1$ elements. Using these two facts we can construct the $\sP$-MinHash as a function of the simplex as illustrated in Figure~\ref{simplexfigure}.

For a probability distribution $x$, mark the point on the unit simplex corresponding to $x$, $(x_1,\ldots,x_{k+1})$, and connect it to each of the corners of the simplex. These edges divide it into $k+1$ smaller simplices that fill the unit simplex. Each internal simplex has volume proportional to the coordinate of $x$ opposite to its unique exterior face. As a result, a uniformly chosen point on the unit simplex will land in one of the sub-simplices with probability given by $x$. $\sP$-MinHashing is equivalent to sampling in this fashion, but holding the chosen point constant when sampling from each new distribution. The match probability is then proportional to the sum of the intersections of simplices that share an external face. This representation makes several properties obvious on small examples that we prove generally in the next section.

\section{Optimality of $\JP$}
\label{jp-jw-bounds}

When MinHashing is used with a key-value store, high collision probabilities are generally more efficient than low collision probabilities, because as we discuss in section \ref{web-document-section}, it is much cheaper to lower them than to raise them. For this reason, we are interested in the question of the highest collision probability that can be achieved through sampling. The constraint that the samples follow each distribution forces the collision probability to remain discriminative, but given that constraint, we would like to make it as high as possible to maximize flexibility and efficiency.

Suppose for two distributions, $x$ and $y$, we want to choose a joint distribution that maximizes $\Pr[H(x)=H(y)]$.  If we were concerned with only these two particular distributions in isolation, the upper bound of $\Pr[H(x)=H(y)]$ is given by the Total Variation distance, or equivalently $1-\text{L}_1(x,y)/2$. Meeting this bound requires the probability mass where $x$ exceeds $y$ to be perfectly coupled with the mass where $y$ exceeds $x$. Both the mass they share and the mass they do not must be perfectly aligned.

Rather than just two, we want to create a joint distribution (or coupling) of all possible distributions on a given space where the collision probability for any pair is as high as possible. It is always possible to increase the collision probability of one pair at the expense of another so long as the chosen pair has not hit the Total Variation limit, so the kind of optimality we are aiming for is Pareto optimality. This requires that no collision probability be able to exceed ours everywhere; any gain on one pair must have a corresponding loss on another.

This by itself would not be a very consequential bound for its retrieval performance. We really only desire high collision probabilities for items that are similar, and we would happily lower the collision probability of a dissimilar pair to increase it for a similar pair.

However we are able to prove something stronger by examining the pair whose collisions must be sacrificed. To increase the collision probability for one pair above its $\JP$, you must always sacrifice collisions on a pair with even higher $\JP$. To get better recall on one pair, you must always give up recall on an even better pair.

The Jaccard index itself is optimal on uniform distributions, and the short proof is a model for the general case.

\begin{theorem} No method of sampling from uniform discrete distributions can have collision probabilities dominating the Jaccard index of their supports. The Jaccard index is Pareto optimal.
\end{theorem}
\begin{proof}
Let $Z$ be the symmetric difference of $X$ and $Y$, $Z = X \cup Y - X \cap Y$, and let $z, x, y$ be the corresponding discrete uniform distributions. The three intersections, $X\cap Y,\  X\setminus Y,\  Y\setminus X$, are disjoint, so the three collision events, $H(x) = H(y)$, $H(z) = H(x)$, and $H(z) = H(y)$, are also disjoint. As disjoint events, their probabilities must sum to at most 1. When the three probabilities are given by the Jaccard index of the corresponding sets, they already sum to 1, respectively, $\frac{|X\cap Y|}{|X\cup Y|}, \frac{|X \setminus Y|}{|X\cup Y|}, \frac{|Y \setminus X|}{|X\cup Y|}$, so the Jaccard index is Pareto optimal.
\end{proof}
To prove the same claim on $\JP$ for all distributions, we need a few tools. We can rearrange $\JP$ to separate the two iteration indices within the max.
\begin{align}
\label{splitindices}
\begin{split}
\JP(x,y) &= \sum_i\frac{x_i}{\sum_j y_j\max\left(\frac{x_j}{y_j}, \frac{x_i}{y_i}\right)}
\end{split}
\end{align}
This lets us characterize the $\argmax$ choices in the denominator as functions of a sorted list of $x_i/y_i$ where $\frac{x_i}{y_i} \geq \frac{x_{i+1}}{y_{i+1}}$. If we reindex $i$ according to this sorted list, then we can describe each $\JP(x,y)_i$ knowing that $i$ terms choose the left side of the max, and $n-i$ terms choose the right. (This also lets us compute $\JP$ in $O(n \log n)$ rather than $O(n^2)$ time by sorting first and keeping running sums of each choice of the max.)

To analyze $\JP$ we will need to refer to each term in its outer sum via subscript. $\JP(x,y)_i \coloneqq \frac{1}{\sum_{j} \max\left(\frac{x_j}{x_i}, \frac{y_j}{y_i}\right)} = z^i_i$. We will also use quantifiers in this subscript to indicate a partial sum, i.e. $\JP(x,y)_{i>a} \coloneqq \sum_{\{i: i>a\}} \JP(x,y)_i$.

\begin{lemma} Useful tools from sorting $\JP$'s indices:  \label{sorted_consequences}
\begin{enumerate}
\item For at least two values of $i$, and distributions $x,y$, $\JP(x,y)_i = \min(x_i, y_i)$.
\item Let $w_1...w_n$ be distributions with disjoint support. Consider two linear combinations of these distributions with coefficients $\alpha$ and $\beta$. $\JP(\alpha \cdot w, \beta \cdot w) = \JP(\alpha, \beta)$
\end{enumerate}
\end{lemma}
\begin{proof}
For a given $i$, if every $\max$ chooses the same side, $\JP(x,y)_i = \min(x_i, y_i)$. $x_i/y_i$ has both an $\argmax$ and $\argmin$ for which this is true. This gives us part 1.

Let $x_k/y_k = x_l/y_l$.
\begin{align*}
\sum_{i\in \{k,l\}}\frac{x_i}{\sum_j y_j\max\left(\frac{x_j}{y_j}, \frac{x_i}{y_i}\right)} = \frac{x_k + x_l}{\sum_j y_j\max\left(\frac{x_j}{y_j}, \frac{x_k}{y_k}\right)} \\
\sum_{j\in \{k,l\}} y_j\max\left(\frac{x_j}{y_j}, \frac{x_i}{y_i}\right) = (y_k + y_l)\max\left(\frac{x_k}{y_k}, \frac{x_i}{y_i}\right)
\end{align*}
Thus, we can form $x'$, $y'$ by merging the mass of $x_k, x_l$ into one element and $y_k, y_l$ into one element and have $\JP(x,y) = \JP(x',y')$ Let $w_1...w_n$ be distributions with disjoint support, and consider $\JP(\alpha\cdot w, \beta\cdot w)$. Repeat the merging process until all elements of each $w_i$ are merged into one. This gives us (2).
\end{proof}

\begin{figure*}
    \centering
    \begin{subfigure}[b]{0.23\textwidth}
        \centering
        \includegraphics[width=0.9\textwidth]{s-xy.png}
        \caption{$\Pr\left[ H(x), H(y) = i, j\right]$}
    \end{subfigure}\hfill
    \begin{subfigure}[b]{0.23\textwidth}
        \centering
        \includegraphics[width=0.9\textwidth]{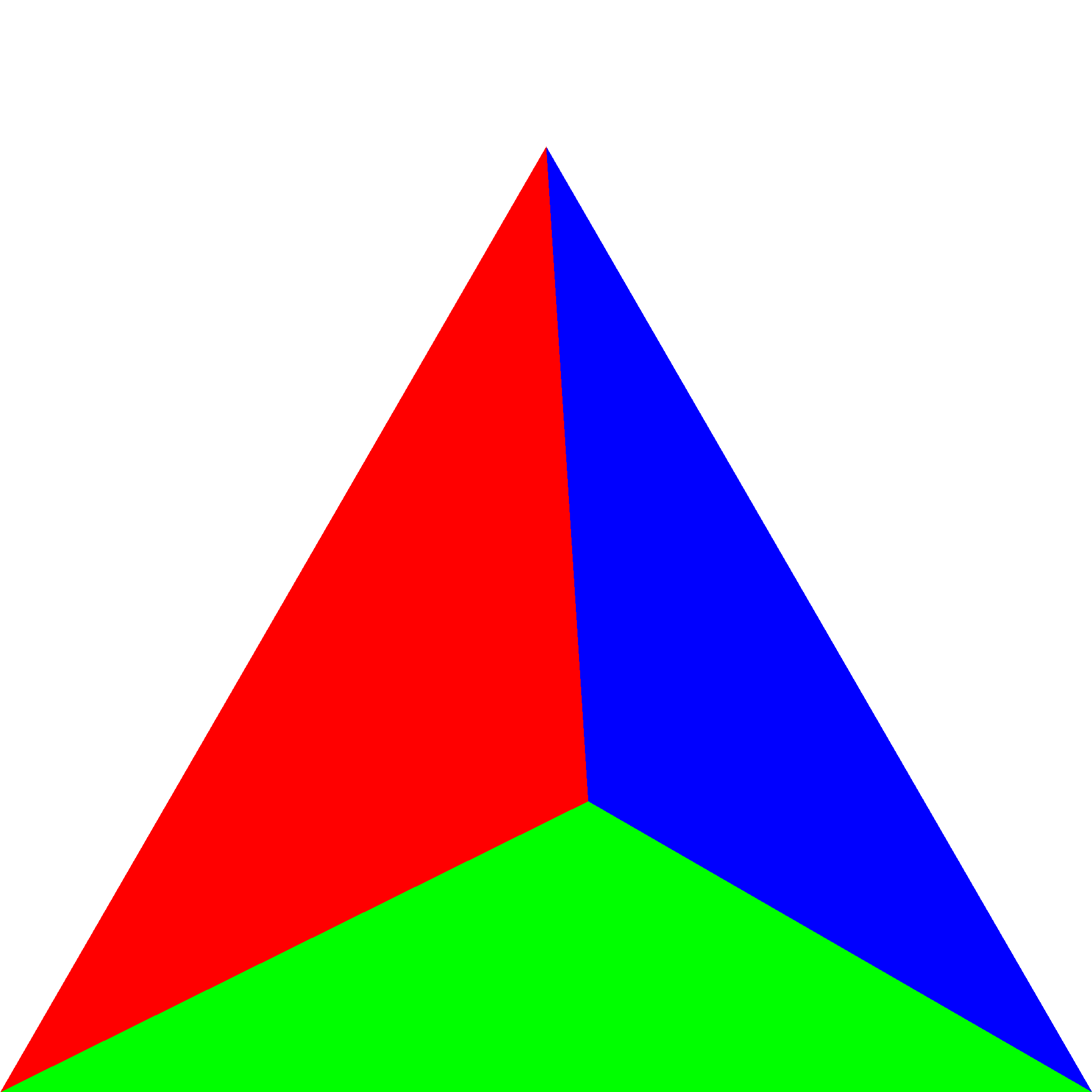}
        \caption{$\Pr\left[ H(z^\text{g}) = i\right]$}
    \end{subfigure}\hfill
    \begin{subfigure}[b]{0.23\textwidth}
        \centering
        \includegraphics[width=0.9\textwidth]{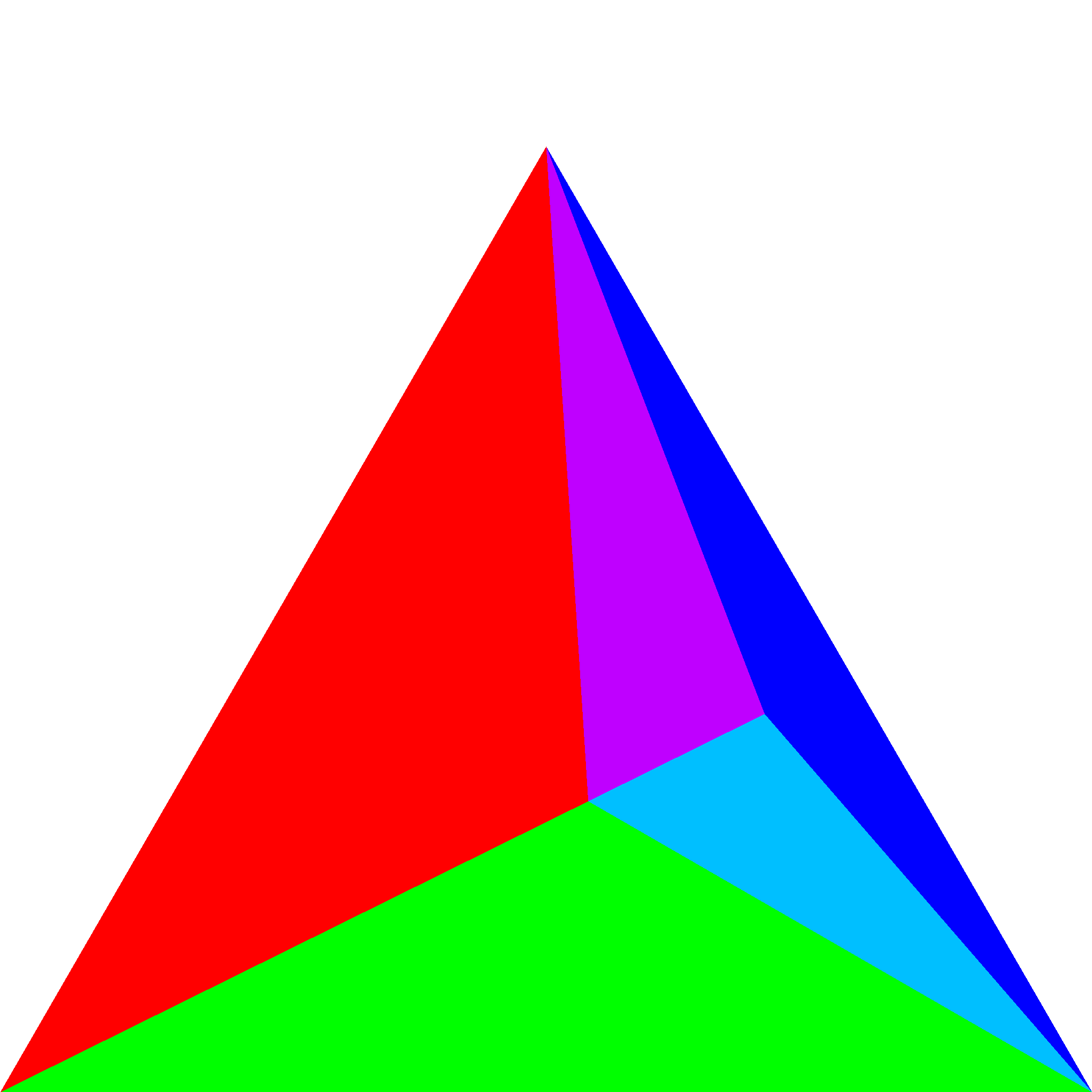}
        \caption{$\Pr\left[ H(x),H(z^\text{g}) = i,j\right]$}
    \end{subfigure}\hfill
    \begin{subfigure}[b]{0.23\textwidth}
        \centering
        \includegraphics[width=0.9\textwidth]{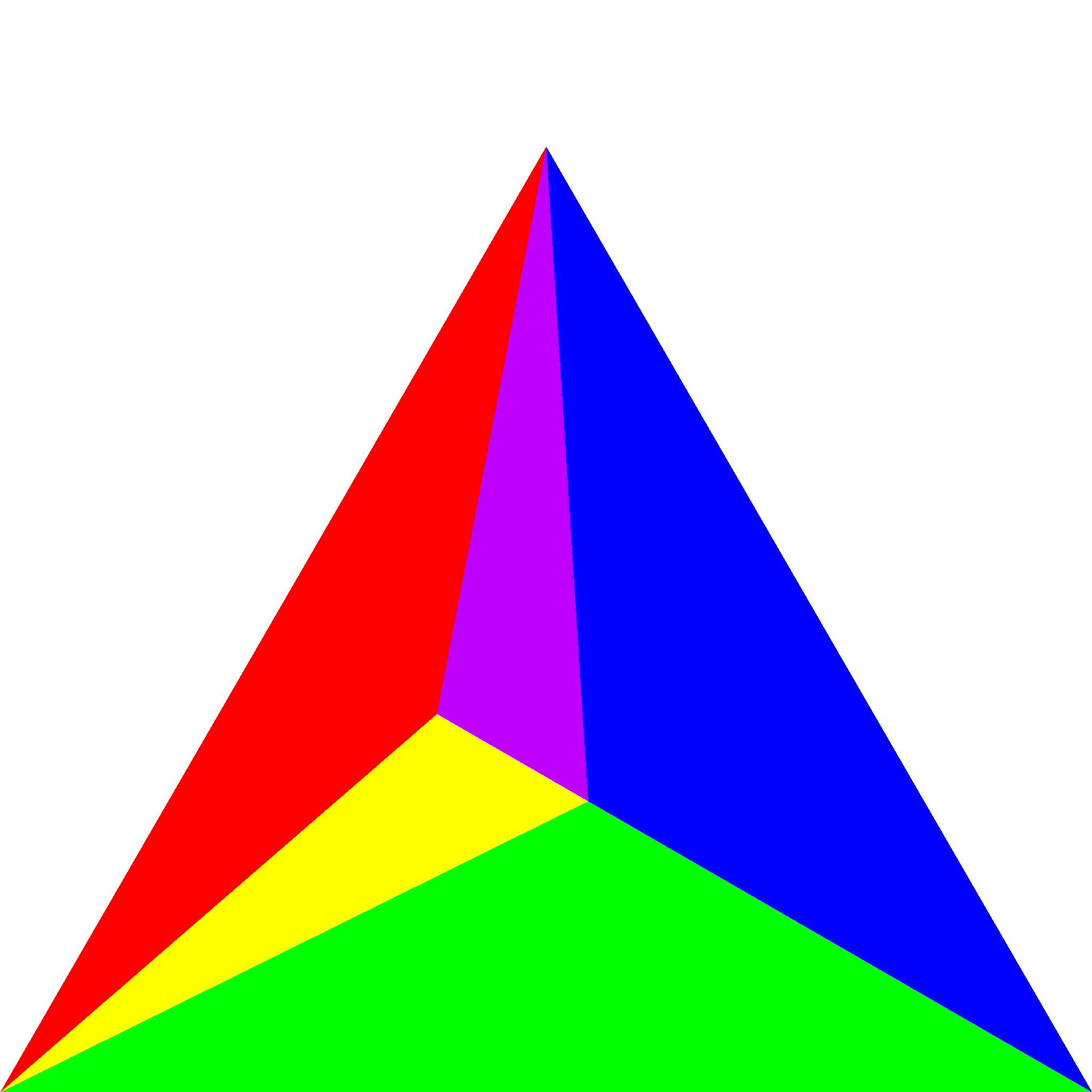}
        \caption{$\Pr\left[ H(y), H(z^\text{g}) = i,j\right]$}
    \end{subfigure}
    \caption{A visual proof of theorem \ref{optimality} on 3 element disributions. Using the same distributions as in Figure \ref{simplexfigure}, the green regions of $x$ and $y$ could be shifted to overlap more and improve the collision probability of this pair, but any modification that achieves that would worsen at least one of collision probabilities between $x$ or $y$ and  $z^\text{green}$ (both of which have higher collision probability than the $(x, y)$ pair.) }
\label{zfigure}
\end{figure*}

This ordering also lets us work more effectively with the $z$ distributions we constructed in theorem \ref{Collision_Probability}. This lemma contains all the algebra needed for the main proof.
\begin{lemma} \label{zterms} For fixed $a$, let $z^a$, $x$, $y$ be probability distributions where $z^a_i \propto \max\left(\frac{x_i}{x_a}, \frac{y_i}{y_a}\right)$.  Reorder $i$ according to the sorting of (\ref{splitindices}) such that $\frac{x_i}{y_i}\geq\frac{x_{i+1}}{y_{i+1}}$. The following are true:
\begin{enumerate}
\item $\forall i \leq a,~ \JP(x,z^a)_i = z^a_i = \min(x_i, z^a_i) \geq z^i_i$.
\item $\forall i \geq a,~ \JP(x,z^a)_i = z^i_i$.
\item $\JP(x,y)_i\nobreak =\nobreak \min(x_i,y_i)\nobreak \implies$ \\
$\JP(x,z^a)_i\nobreak =\nobreak \min(x_i, z^a_i)$.
\item $\JP(x, z^a) \geq \JP(x,y)$ and $\JP(y, z^a) \geq \JP(x,y)$
\end{enumerate}
\end{lemma}
\begin{proof} For $i \leq a$, $\frac{x_i}{y_i} \geq \frac{x_a}{y_a}  \Rightarrow \frac{x_i}{x_a}\geq \frac{y_i}{y_a}$, hence $z_i^a \propto \frac{x_i}{x_a}$. Similarly, for $i \geq a$, $z^a_i \propto \frac{y_i}{y_a}$. Working first with the lower group of indices, 
\begin{align*}
\forall i \leq a,~ \JP(x,z^a)_i &= \frac{1}{\sum_j \max\left(\frac{x_j}{x_i}, \frac{\max\left(\frac{x_j}{x_a}, \frac{y_j}{y_a}\right)}{x_i/x_a}\right)}\\
&= \frac{x_i/x_a}{\sum_j \max\left(\frac{x_j}{x_a}, \frac{y_j}{y_a}\right)} = z^a_i
\end{align*}
And since $\frac{x_i/x_a}{\sum_j \max\left(\frac{x_j}{x_a}, \frac{y_j}{y_a}\right)} \leq \frac{x_i/x_a}{\sum_j \frac{x_j}{x_a}}$ we also know that $\forall i \leq a,~ z^a_i = \min(x_i, z^a_i)$. Furthermore, $\forall i \le a$, 
\begin{align*}
z^a_i &= \frac{1}{\sum_j \max\left(\frac{x_j/x_a}{\max\left(\frac{x_i}{x_a}, \frac{y_i}{y_a}\right)},\frac{ y_j/y_a}{\max\left(\frac{x_i}{x_a}, \frac{y_i}{y_a}\right)}\right)}\\
&= \frac{1}{\sum_j \max\left(\min\left(\frac{x_j}{x_i}, \frac{x_jy_a}{x_ay_i}\right), \min\left(\frac{y_j}{y_i}, \frac{y_j x_a}{y_ax_i}\right)\right)} \geq z^i_i
\end{align*}
which gives us part 1. Now, continuing on to the upper group of indices,
\begin{align*}
\forall i \geq a,~ \JP(x,z^a)_i &= \frac{1}{\sum_j \max\left(\frac{x_j}{x_i}, \frac{\max\left(\frac{x_j}{x_a}, \frac{y_j}{y_a}\right)}{y_i/y_a}\right)}\\
&= \frac{1}{\sum_j \max\left(\frac{x_j}{x_i}, x_j\frac{y_a}{x_a,y_i}, \frac{y_j}{y_i}\right)}
\end{align*}
Since $\forall i \geq a,~\frac{x_ay_i}{y_a} \geq x_i$ we can simplify further to conclude part 2. $\forall i \geq a,~ \JP(x,z^a)_i = z^i_i$. By noting that the choices within the max are preserved, we conclude part 3. Finally, having bounded all indices, part 4:
\begin{align*}
\forall i,~ \JP(x, z^a)_i &\geq \JP(x,y)_i\\
\JP(x,z^a) &\geq \JP(x,y). \mbox{\qedhere}
\end{align*}
\end{proof}
We now have the tools to prove the optimality of $\JP$.

\begin{theorem}\label{optimality} Let $G$ be a sampling method. If $\Pr[G(x) = G(y) = i] > \JP(x,y)_i$, there exists a $z$ such that $\Pr[G(x) = G(z)] < \JP(x,z)$ or $\Pr[G(y) = G(z)] < \JP(y,z)$. This implies that no method of sampling from discrete distributions can have collision probabilities dominating $\JP$. $\JP$ is Pareto optimal. 

Furthermore, $\JP(x,z) \geq \JP(x,y)$ and $\JP(y,z) \geq \JP(x,y)$. To exceed $\JP(x,y)$, $G$ must sacrifice at least one pair that is closer under $\JP$ than $(x,y)$.\end{theorem}
\begin{proof}

Let $m$ be the number of elements $i$ for which $\JP(x,y)_i < \min(x_i,y_i)$.

In the base case where $m = 0$, $\JP(x,y) = \sum_i \min(x_i,y_i)$ which cannot be improved.

Assume the proposition to be proved is true $\forall x, \forall y$, and $\forall p < m$. By \ref{sorted_consequences}.1 we know that $ m \le n - 2$, since at least the two endpoints of the sorted list have reached their upper bound. We proceed by induction on $m$.

As in theorem \ref{Collision_Probability}, for each $a$ where $\JP(x,y)_a < \min(x_a,y_a)$, construct a distribution $z^a$, where $z^a_i \propto \max\left(\frac{x_i}{x_a}, \frac{y_i}{y_a}\right)$. Reorder $i$ according to the sorting of (\ref{splitindices}) such that $\frac{x_i}{y_i}\geq\frac{x_{i+1}}{y_{i+1}}$. From lemma \ref{zterms} we know $\forall i \leq a,~ \JP(x,z^a)_i = z^a_i = \min(x_i, z^a_i)$ and $\forall i \geq a,~ \JP(y,z^a)_i = z^a_i = \min(y_i, z^a_i)$.

Now consider a new sampling method, $G$. The following events are pairwise disjoint by inspection. 
\begin{align*}
E_a^< &\coloneqq\{G(x) = G(z^a) < a\} \\
E_a^= &\coloneqq\{G(x) = G(y) = a\} \\
E_a^> &\coloneqq\{G(y) = G(z^a) > a\}
\end{align*}
So their probabilities are constrained. $\Pr[E_a^=] + \Pr[E_a^<] + \Pr[E_a^>] \le 1$. When these probabilities are given by $\JP$ they already sum to 1. $\Pr[E_a^=] + \Pr[E_a^<] + \Pr[E_a^>] =  z^a_a + \sum_{i<a} z^a_i + \sum_{i>a} z^a_i = 1$.

From this bound, if $\Pr[G(x) = G(y) = a] > \JP(x,y)_a$, at least one of
\begin{align*}
\Pr[G(x) = G(z^a) < a] < \JP(x,z^a)_{i<a} \\
\Pr[G(y) = G(z^a) > a] < \JP(y, z^a)_{i>a}
\end{align*}
must be true. Since the two cases are symmetric, we will assume the first one: $\Pr[G(x) = G(z^a) < a] < \JP(x,z^a)_{i<a}$.

If $\Pr[G(x) = G(z^a)] < \JP(x,z^a)$ then our $z$ is $z^a$ and we are done. Otherwise, $\Pr[G(x) = G(z^a)] \geq \JP(x,z^a)$, and the terms $i\geq a$ must compensate for the loss on the terms $i<a$, so $\Pr[G(x) = G(z^a) \geq a] > \JP(x,z^a)_{i\geq a}$.

However, $\JP(x, z^a)_a = \JP(y,z^a)_a = z^a_a$, so these terms have exhausted $z^a$ and cannot be increased. Using \ref{zterms}.1 and \ref{zterms} we know that this adds at least one additional term that is fully consumed, i.e. the size of  $\{i : \JP(x, z^a)_i < \min(x_i,z^a_i)\}$ is less than $m$. So by induction hypothesis, if $\Pr[G(x) = (z^a) < a] > \JP(x,z^a)_{i<a}$ there is a new $z$ for which $\Pr[G(x) = G(z)] < \JP(x,z)$ or $\Pr[G(z^a) = G(z)] < \JP(z^a,z)$.

By \ref{zterms} we know that $\JP(x, z^a) \geq \JP(x,y)$ and by induction we know $\JP(x, z) \geq \JP(x, z^a)$, so we conclude $\JP(x, z) \geq \JP(x, y)$ and symmetrically $\JP(y, z) \geq \JP(x, y)$.
\end{proof}

Figure \ref{zfigure} shows the mechanism of the proof intuitively. On three element distributions, two of the elements are fully constrained, in this case the blue and red terms, so we construct our adversarial $z$ around green. On three elements, no induction is necessary and the diagram itself proves the relationship.

Since $\JP$ is only Pareto optimal, we should be able to find a sampling method that exceeds it for some pairs but is below it on others. We can generalize our algorithm to construct such a method. Consider arranging the elements of the state space as the leaves of a tree. Internal nodes in the tree are given the weight of the sum of their children, and each is assigned its own exponential hash. Perform $\sP$-MinHash among the children of the root node. If the selected node is a leaf, emit it as the sample. If it is an internal node, recurse and repeat. In this generalization, our original algorithm is represented by making all elements direct children of the root. We can prioritize collisions on an index $i$ by placing it closer to the root node than all others; in particular, the tree $(i, ( 1, \ldots, n ))$. Since $i$ and its internal-node sibling form a two element distribution, by lemma \ref{splitindices} the probability of a collision on $i$ will be $\min(x_i, y_i)$ for all $x,y$.

What of $\JW$ then? Is it another Pareto optimum? It is not, it is dominated by $\JP$.

\begin{theorem} \label{Jaccard_bounds}
If $\JW(x,y)=\frac{1-p}{1+p}$, then $\frac{1-p}{1+p} \leq \JP(x,y)\leq 1-p$, and there are distributions that achieve both bounds on $\JP$ for any value of $\JW$.
\end{theorem}
\begin{proof}
The lower bound becomes clear by rewriting $\JW$ in a similar form to $\JP$.
\begin{align*}
\JW(x, y) &= \sum_{i} \frac{1}{\sum_{j} \frac{\max\left(x_j, y_j\right)}{\min\left(x_i, y_i\right)}}\\
\sum_{i} \frac{1}{\sum_{j} \max\left(\frac{x_j}{x_i}, \frac{y_j}{y_i}\right)} &\geq \sum_{i} \frac{1}{\sum_{j} \max\left(\frac{x_j}{x_i}, \frac{y_j}{y_i}, \frac{x_j}{y_i}, \frac{y_j}{x_i}\right)}\\
\JP(x,y) &\geq \JW(x,y)
\end{align*}
To achieve this lower bound, we can transform the distributions by moving the ``excess" mass to new elements.
\begin{gather*}
	x'_{2i} = y'_{2i} = \min\left(x_i, y_i\right)\\
    x'_{2i+1} = \max\left(x_i - y_i, 0\right), \quad y'_{2i+1} = \max\left(y_i - x_i, 0\right)
\end{gather*}
Shifting the mass in this way has no effect on $\JW$, but it decreases $\JP$ to equal $\JW$. $x'$ and $y'$ can be expressed as linear combinations over the three sets of indices, so using lemma \ref{sorted_consequences}.2, $\JP(x',y') = \JP\left((0, 1-p, p),(p, 1-p, 0)\right) = \frac{1-p}{1+p}$

To achieve the upper bound, $1-p$, consider inverting this transformation, reallocating the $p$ extra mass to maximize $\JP(x'',y'')$ while holding $\JW(x'',y'')$ constant. To avoid increasing $\JW$, we must add the mass to disjoint elements, so divide the indices into two sets, $X,Y$. We find that if we distribute the mass proportional to the original value, $\JP$ reaches the total variation limit regardless of the choice of $X,Y$. Let $|X| \coloneqq \sum_{i\in X}\min(x_i, y_i)$.
\begin{align*}
\forall i \in X,\enskip x''_i &= \min\left(x_i, y_i\right) \frac{|X|+p}{|X|},\enskip y''_i = \min\left(x_i, y_i\right) \\
\forall i \in Y,\enskip y''_i &= \min\left(x_i, y_i\right) \frac{|Y|+p}{|Y|},\enskip x''_i = \min\left(x_i, y_i\right)
\end{align*}
We can express this as a linear combination of two distributions with disjoint support. $\JP(x'',y'') = \JP\left((|X|+p,|Y|), (|X|,|Y|+p)\right) = 1-p$. Since $p$ is the total variation distance of $x$ and $y$, $1-p$ is the maximum collision probability that is possible between two distributions in any context, so it is the upper bound here as well.
\end{proof}
This gives us some insight into how $\JP$ and $\JW$ differ. $\JP$ ranks distributions as more similar than $\JW$ if their extra mass is on elements that both distributions share. 

Like $1 - \JW$, $1 - \JP$ is a metric on probability distributions.

\begin{theorem} $1 - \JP$ is a proper metric on $\mathcal{P}$ where $\mathcal{P}(\Omega)$ is the space of probability distributions over a finite set $\Omega$.
\end{theorem}
\begin{proof}
Symmetry is obvious. Non-degeneracy over $\mathcal{P}$ follows from $\frac{1}{\sum_{j \in \Omega} \max\left(\frac{x_j}{x_i}, \frac{y_j}{y_i}\right)} \le \min(x_i, y_i)$. The triangle inequality follows from being a collision probability, $1 - \JP(x,y) = \Pr\left[H(x) \neq H(y)\right]$. Therefore, 
\begin{align*}
\Pr\left[H(x) = H(y)\right] \ge \Pr\left[H(x) = H(z) \land H(y) = H(z)\right] \\
\Pr\left[H(x) \neq H(y)\right] \le \Pr\left[H(x) \neq H(z) \lor H(y) \neq H(z)\right]
\end{align*}
for any distribution $z$. But by the union bound,
\begin{align*}
\begin{split}
\Pr\big[&H(x) \neq H(z) \lor H(y) \neq H(z)\big] \\
&\le \Pr\big[H(x) \neq H(z)\big] + \Pr\big[H(y) \neq H(z)\big]\quad\qedhere
\end{split}
\end{align*}
\end{proof}

\section{Hashing on Dense and Continuous Data}

\begin{algorithm}[t]
\DontPrintSemicolon
\SetKwInOut{Input}{input}\SetKwInOut{Output}{output}
\SetKwRepeat{Do}{do}{while}
\Input{sample space $\Omega$, \\ sigma-finite measure $\mu$,\\ \textit{proposal} sigma-finite measure $\lambda$, \\ finite upper bound, $B \coloneqq \max(\mu(i)/\lambda(i))$\\ shared random seed $s$}
\SetKw{KwWhere}{where}
\Output{Stable sample from $(\Omega, \mathcal{F}, \mu)$}
$(M, X^*, k, e_{-1}) \leftarrow (\inf, \text{null}, 0, 0)$\;
\Do{$M > e_k / B$}{
$e_k \leftarrow - \log\left(h(k, s)\right)/\lambda(\Omega) + e_{k-1} $\;
$X_k \leftarrow \text{Sample}(\lambda(\Omega), s)$\;
$M_k \leftarrow e_k \lambda(X_k)/\mu(X_k)$\;
\If{$M_k < M$} {
$M \leftarrow M_k$\;
$X^* \leftarrow X_k$\;
}
$\Omega \leftarrow \Omega \setminus X_k$\;
$k \leftarrow k+1$\;
}
\Return{$X^*$}
\caption{Dense and Continuous $\sP$-MinHash. ``Global-Bound" A* Sampling \cite{astar} with a fixed seed.}
\label{p-minhash-algorithm-dense}
\end{algorithm}
The algorithm we've presented so far is suitable for sparse data such as documents or graphs. It is linear in the number of non-zeros, equivalent to \citeauthor{wminhash} \citeyear{wminhash} \cite{wminhash}. On dense data (such a image feature histograms\cite{denseminhash}) there's significant overlap in the supports of each distribution, so rehashing each element for every distribution wastes work. With a shared stream of sorted hashes, we expect the hash we select for each distribution to be biased towards the beginning of the stream, and closer to the beginning when the data is denser. Therefore one might expect that we could improve performance by searching only some prefix of the stream to find our sample.

A* Sampling (\citeauthor{astar} \citeyear{astar}\cite{astar}) explores this idea thoroughly, and we lean on it heavily in this section. In particular, we use their ``Global Bound'' algorithm, and essentially just run it with a fixed random seed (algorithm \ref{p-minhash-algorithm-dense}.) We leave the proof of running time and of correctness as a sampling method to that work, and limit our discussion to the proof of the resulting collision probability. (Their derivation uses Gumbel variables and maxima. We use exponential variables and minima to make the continuity with the rest of our work clear, which is achieved by a simple change of variables.)

The key insight of A* Sampling is that when a (possibly infinite) stream of independent exponential random variables is ordered, the exact marginal distributions of each variable can be computed as a function of the rank and the previous variables. If the vector of sorted exponential variables is $e$ with corresponding parameter vector $x$, then once $e_1,\ldots,e_{k-1}$ are all known, the distribution of $e_k$ is a truncated exponential with rate $|x \setminus \bigcup_{i=1}^{k-1} x_i|$ truncated from below at $e_{k-1}$. The ``statelessness" of exponential variables makes this truncation easy to accomplish. Simply generate the desired exponential, and add $e_{k-1}$ to shift it.

\begin{figure*}
\hfill
    \begin{subfigure}[t]{0.40358391608\textwidth}
        \centering
        \includegraphics[width=\textwidth]{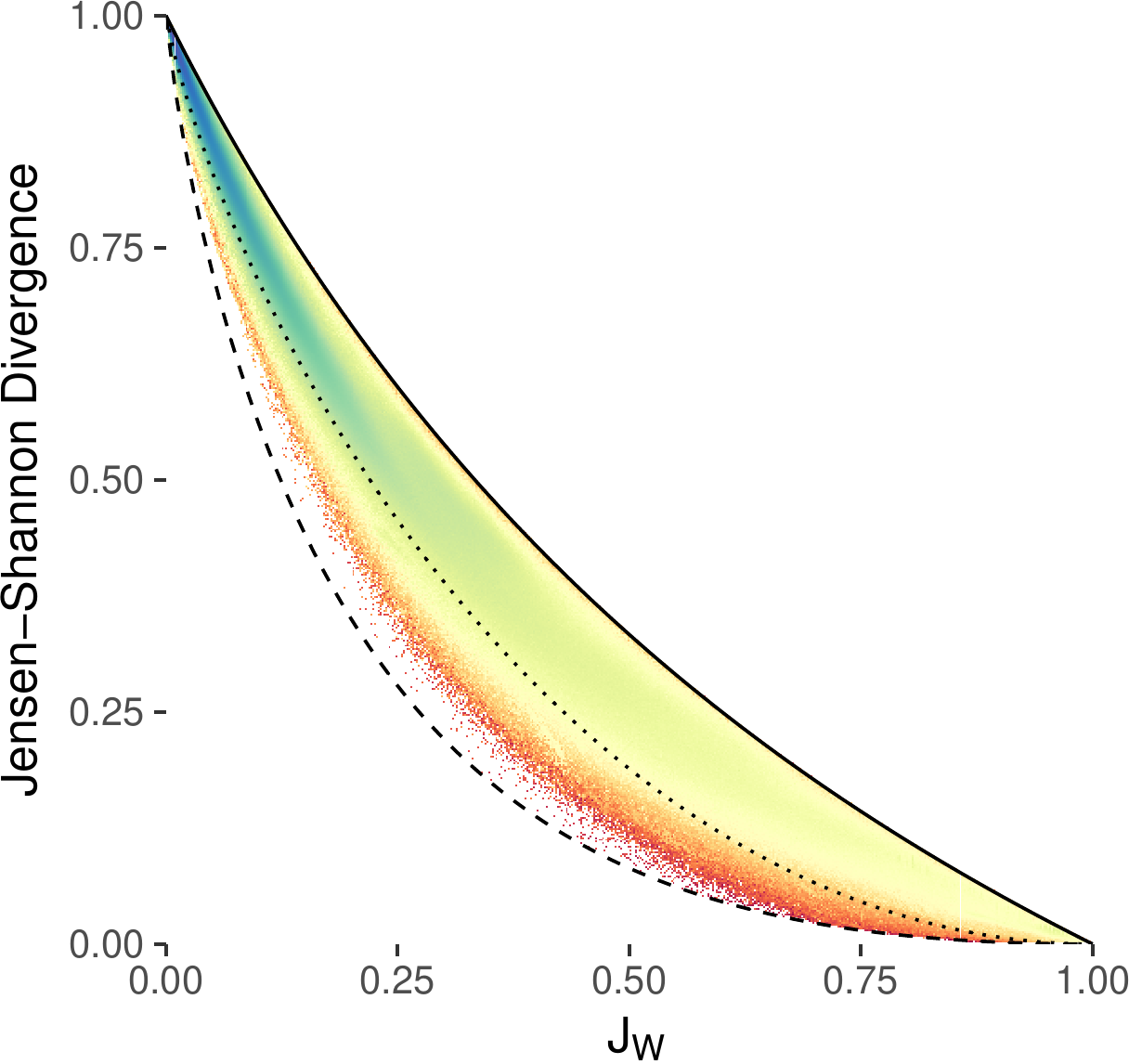}
    \end{subfigure}
    \hfill
    \begin{subfigure}[t]{0.4949300699\textwidth}
        \centering
        \includegraphics[width=\textwidth]{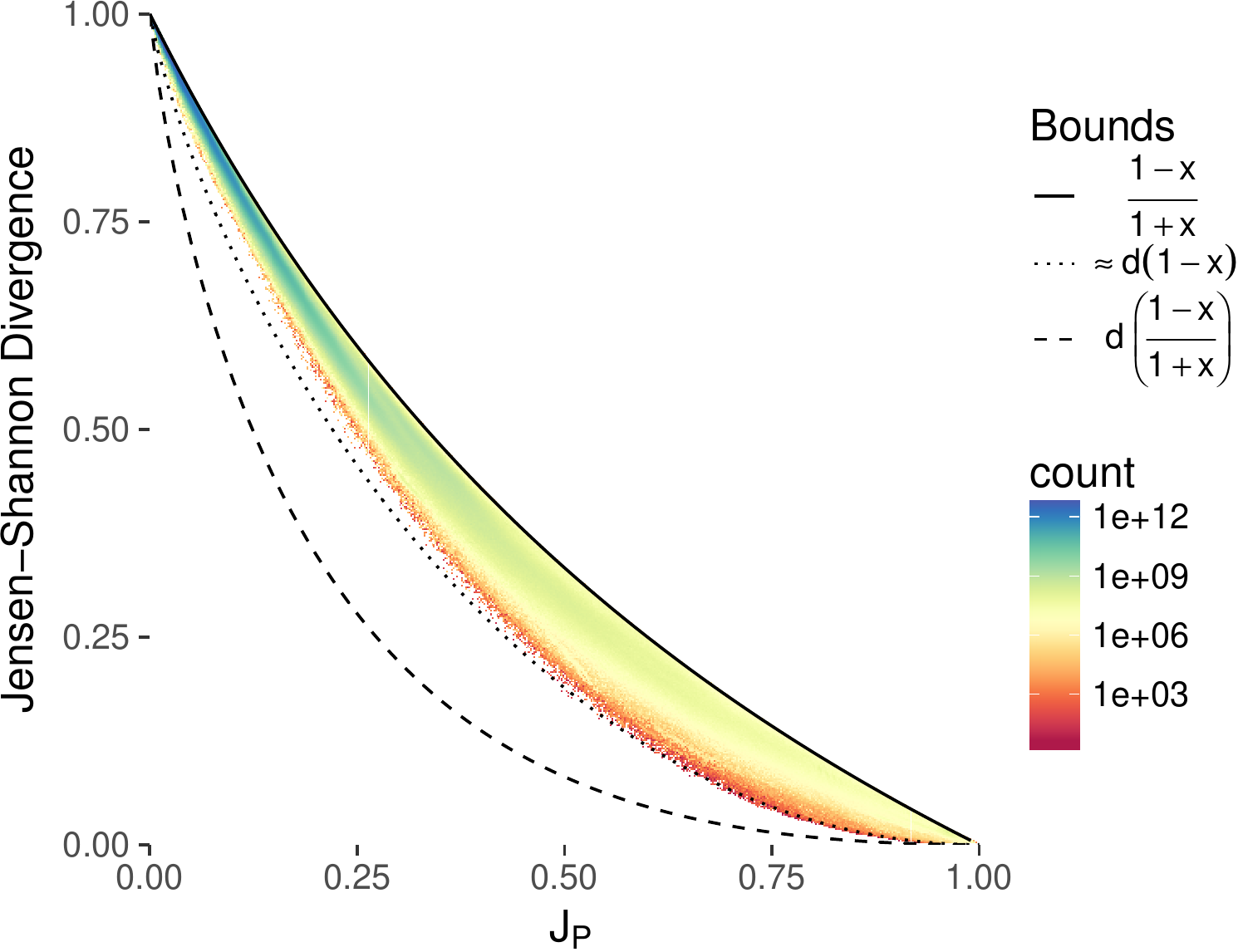}
    \end{subfigure}
    \hfill
    \caption{The Jensen-Shannon divergence compared with $\JP$ and $\JW$ on pairs of normalized unigram term vectors of random web documents. JSD has a much tighter relationship with $\JP$ than $\JW$. We show exact bounds for $\JW$ against JSD and approximate bounds for $\JP$ against JSD where $d(x) = \frac{(1-x)}{2}\log_2(1 - x) + \frac{(1 + x)}{2}\log_2(1 + x)$. The curve that appears to lower bound JSD against $\JP$ is violated on $10^{-7}$ of the pairs.}
\label{jsd-dist}
\end{figure*}

\begin{figure*}
    \centering
    \begin{subfigure}[t]{0.475\textwidth}
        \centering
        \includegraphics[width=\textwidth]{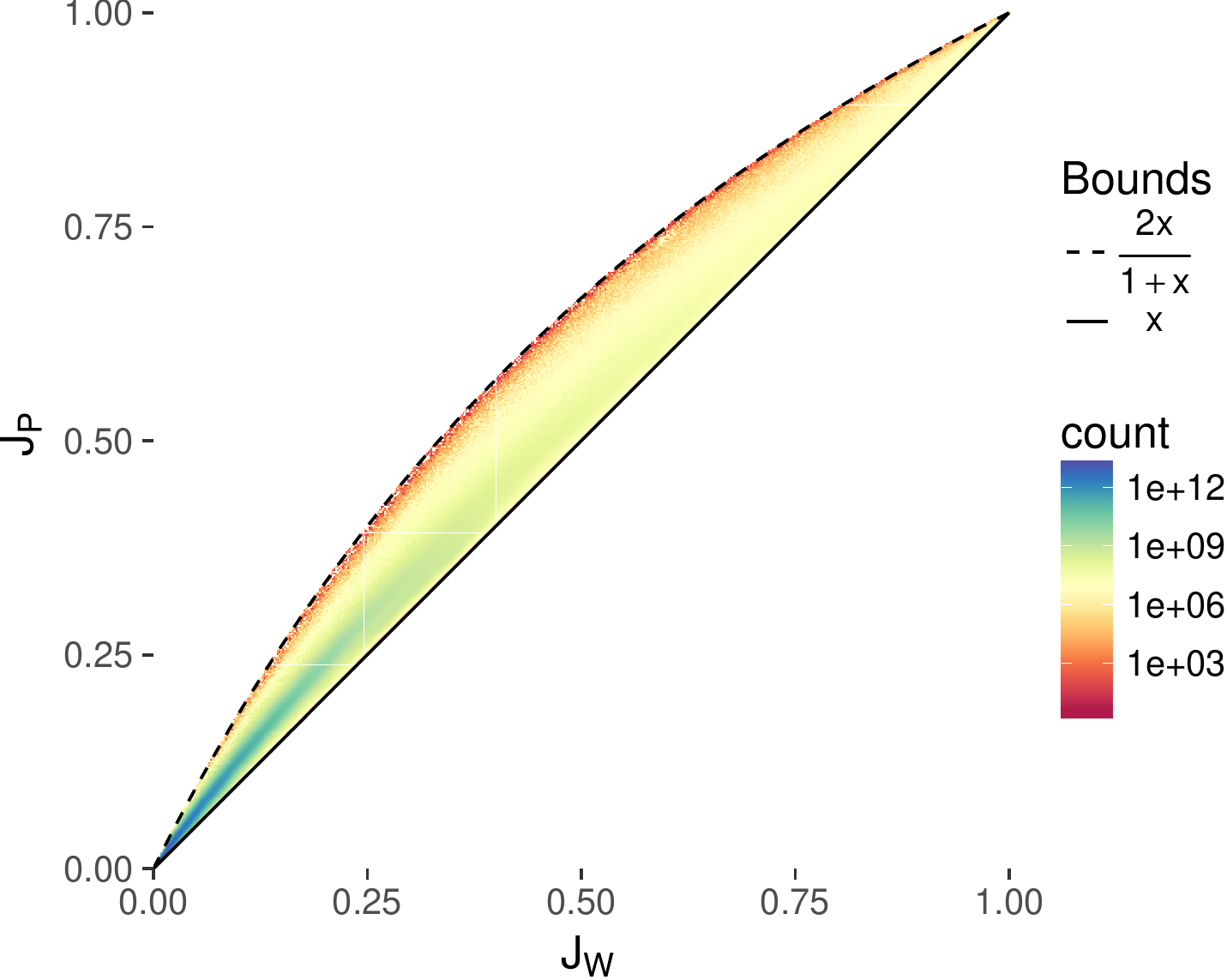}
    \end{subfigure}\hfill
    \begin{subfigure}[t]{0.47416958042\textwidth}
        \centering
        \includegraphics[width=\textwidth]{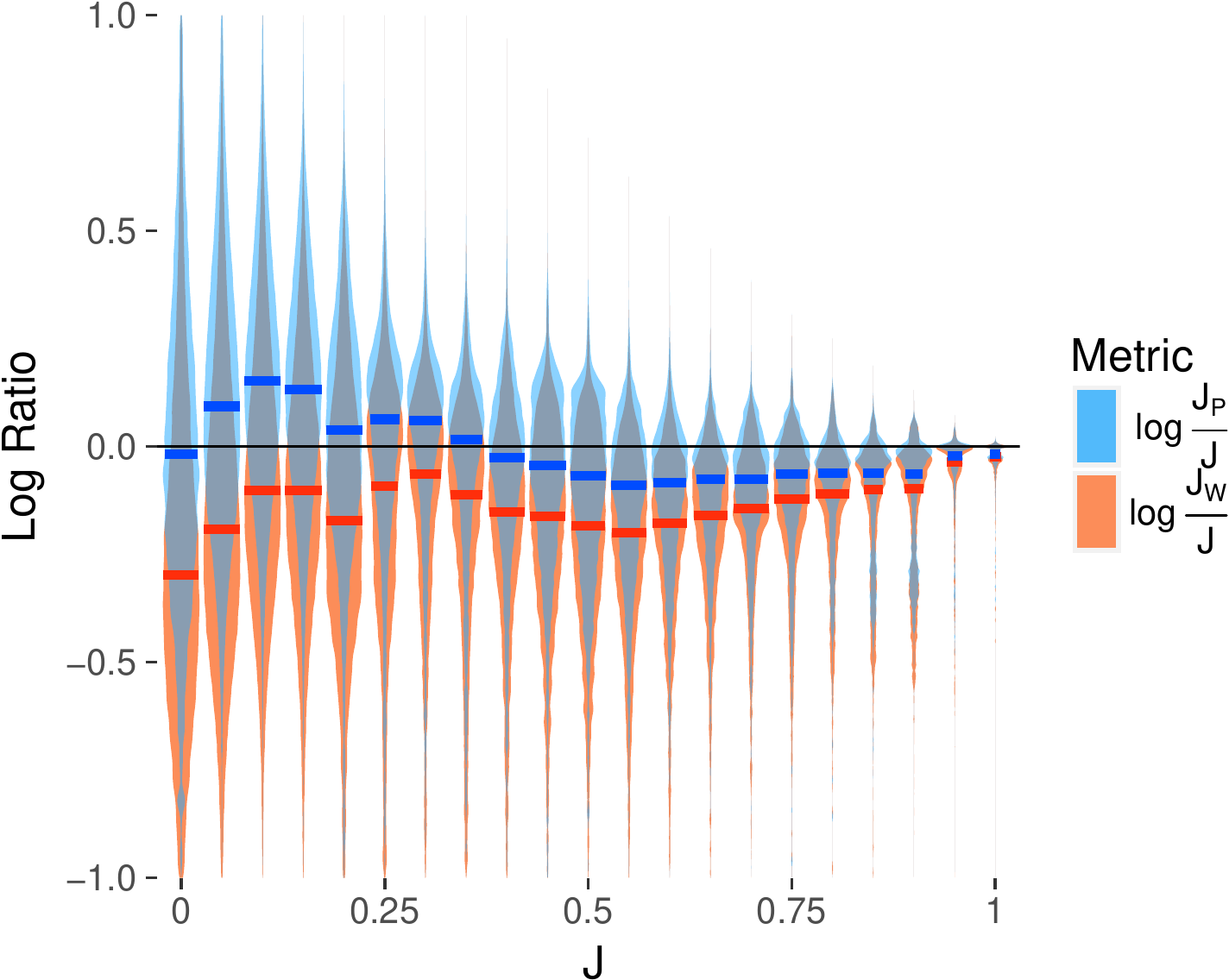}
    \end{subfigure}
    
    \caption{Comparison of Jaccard measures on normalized unigram term vectors of pairs of random web documents. The left graph shows the joint distribution of $\JP$ and $\JW$ and the bounds we prove. The right graph shows the conditional distributions of $\JP$ and $\JW$ against the (set) Jaccard index of the terms. We show the distribution of the log of their ratios against the Jaccard index and highlight the median. $\JP$ is generally centered around the Jaccard index, while $\JW$ is consistently centered below, as predicted by their behavior on uniform distributions.}
    \label{jp-jw-dist}
\end{figure*}

\begin{figure*}
    \centering
    \begin{subfigure}[t]{0.49\textwidth}
        \centering
		Precision/Recall on $\text{JSD}< 0.25$
		\par\medskip
        \includegraphics[width=\textwidth]{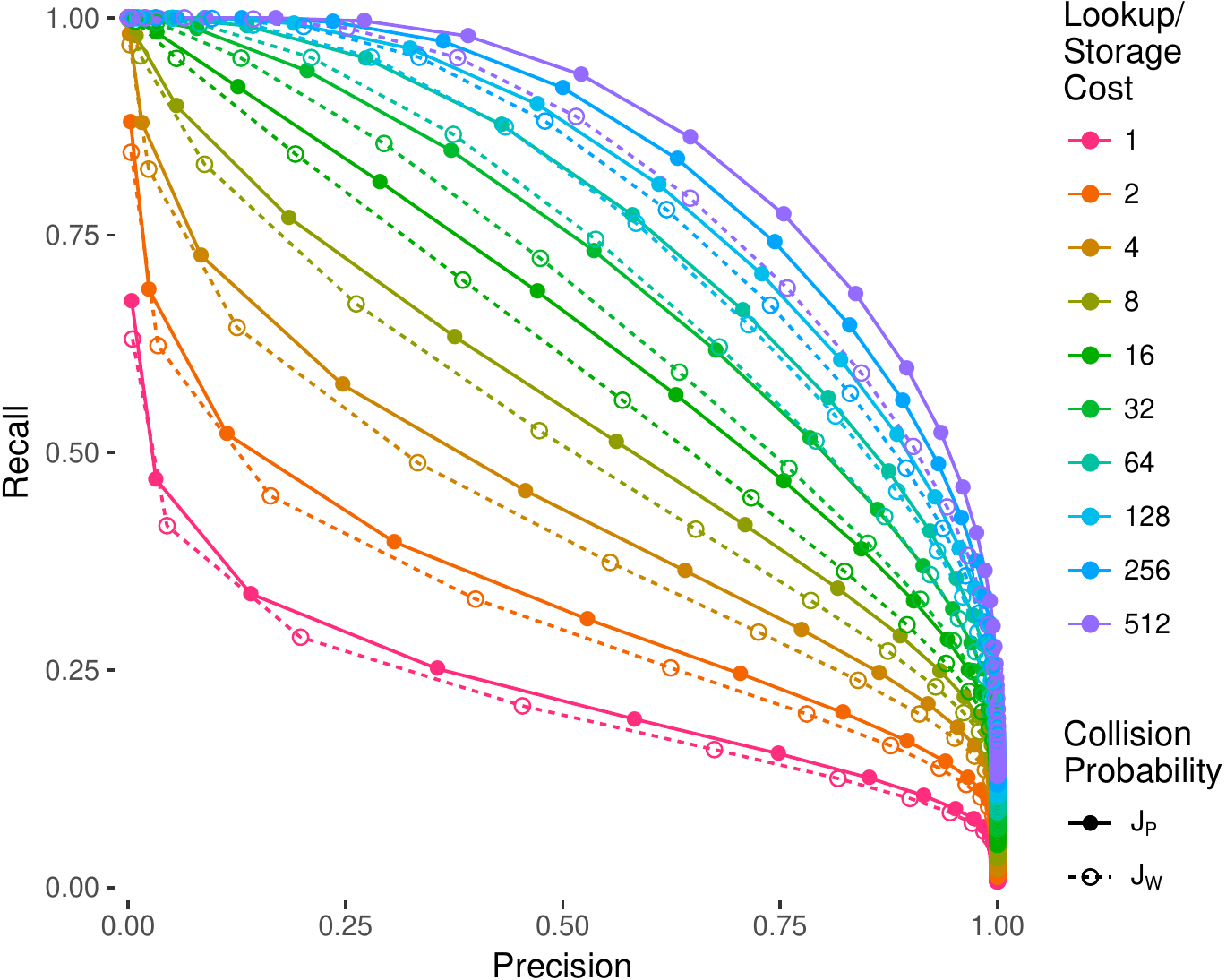}
        \caption{Performance on retrieving pairs with Jensen-Shannon divergence less than 0.25. $\JP$ performs better both by having a tighter relationship with JSD, and by having higher match probabilities overall, making high recall cheaper to achieve. $\sP$-MinHash achieves the same performance with 64 hashes that a $\sW$-MinHash achieves with 128.}
    \end{subfigure}\hfill
    \begin{subfigure}[t]{0.49\textwidth}
        \centering
        Precision/Recall on $\JW > 0.5$
        \par\medskip
        \includegraphics[width=\textwidth]{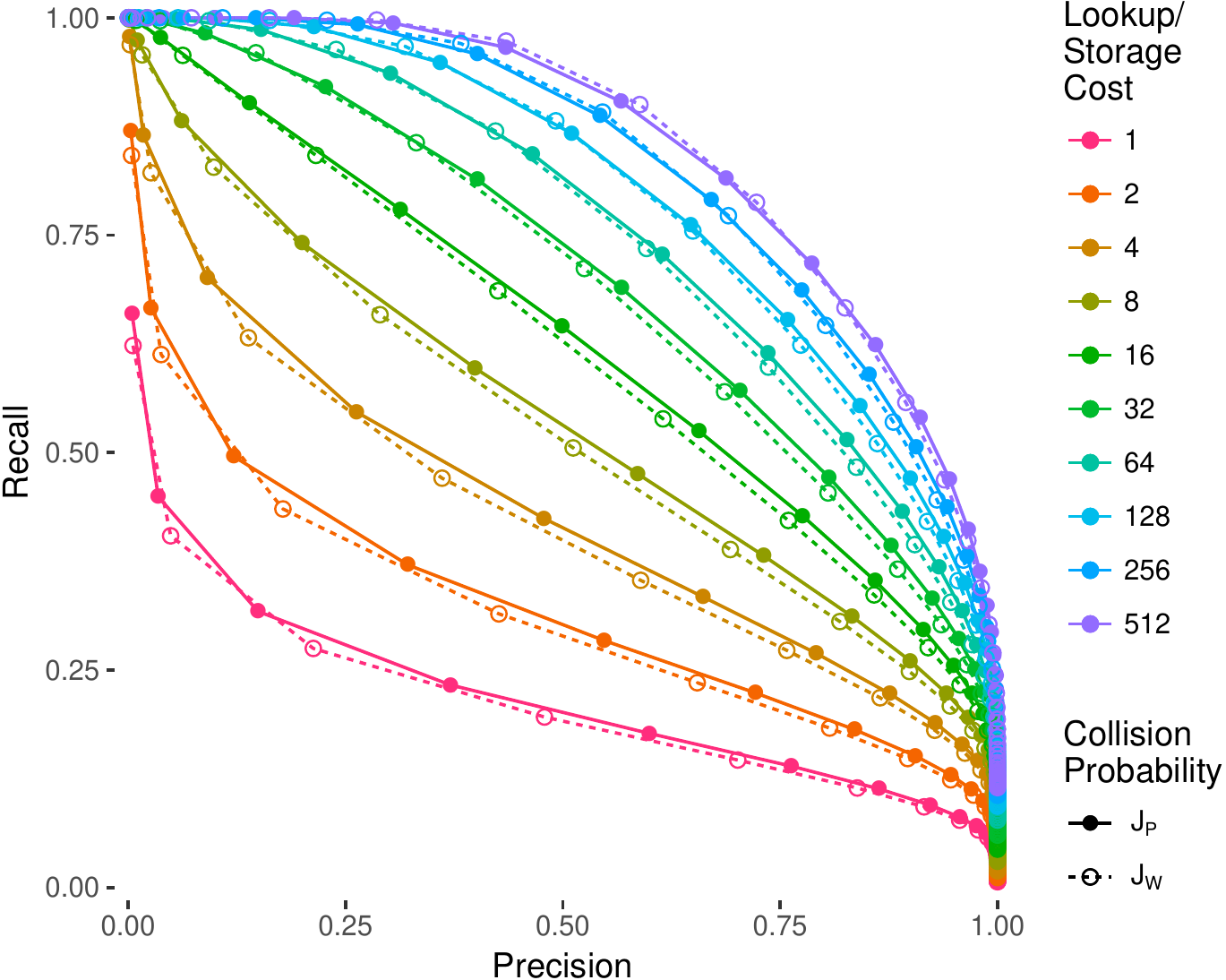}
        \caption{Performance on retrieving pairs with $\JW$ greater than 0.5. The fact that  $\sP$-MinHash slightly outperforms $\sW$-MinHashes illustrates the benefit of higher match probabilities. $\JP$ is higher than $\JW$ for pairs that are similar under $\JW$, which makes high recall cheaper. As the cost increases, $\JW$ overtakes $\JP$ as the difference between the two scores becomes the larger factor. }
    \end{subfigure}
    
    \caption{Precision/recall curves illustrating the typical case of retrieval using a key-value store. Each point represents outputting $o$ independent sums of $a$ hashes each, for a collision probability of $1-(1-p^a)^o$. The cost in storage and CPU is dominated by $o$, so we connect these points to show the trade-offs possible at similar cost.}
    \label{precrec}
\end{figure*}

To change the parameters of those exponentials and find the new minimum element, only a small prefix of the list must be examined. The running time of finding the new minimum is a function of the difference between the two vectors of parameters, and has equivalent running time to rejection sampling. This gives algorithm \ref{p-minhash-algorithm-dense} equivalent running time to the state of the art for computing a MinHash of dense data, \citeauthor{denseminhash} \citeyear{denseminhash}\cite{denseminhash}. In addition to admitting an unbounded list of elements, algorithm \ref{p-minhash-algorithm-dense} is also applicable to continuous distributions, as the paper\cite{astar} describes in detail. 

Let's first show that algorithm \ref{p-minhash-algorithm-dense} gives $\JP(\mu, \nu)$ when $\Omega$ is \emph{finite}. Indeed, this construction is simply an alternative way of finding the minimum $-\log U_i / \mu_i$. A* sampling merely reads off the minimum of $-\log U_i /\lambda_i * \lambda_i/ \mu_i = -\log U_i / \mu_i$, and similar for $\nu$. 

To prove the general case for infinite $\Omega$, we need to first define what we mean by $\JP(\mu, \nu)$ in that setting. One option is to replace all the summation by integrals in the formula ~\eqref{JP_double_sum_formula}. This runs into two difficulties however:
\begin{enumerate}
\item A probability space $\Omega$ may not be a subset of $\R^n$.
\item Either $\mu$ or $\nu$ could be singular.
\end{enumerate}

Instead, we define it as a limit over increasingly finer finite partitions of $\Omega$. More formally,
\begin{definition}
Assume $\JP(\mu, \nu)$ is defined as before when $|\Omega| < \infty$, we define 
\begin{align}
\JP(\mu, \nu) = \inf_{\mathcal{F} \vdash \Omega} \JP(\mu_\mathcal{F}, \nu_\mathcal{F}),
\end{align}
where $\mathcal{F}$ ranges over finite partitions of the space $\Omega$, and $\mu_\mathcal{F}$ denotes the push-forward of $\mu$ with respect to the map $\pi: \Omega \to \mathcal{F}$, $\pi(x) = Q \in \mathcal{F}$ iff $x \in Q$. ($\mu_\mathcal{F}$ is simply a coarsified probability measure on the finite space $\mathcal{F}$ where it (tautologically) assigns probability $\mu(Q)$ to the element $Q \in \mathcal{F}$.)
\end{definition}
First we verify that the definition above coincides with $\JP$ when $\Omega$ is finite. 
\begin{lemma}
Let $|\Omega'| +1 = |\Omega|  = n$, $\mu, \nu \in \mathcal{P}(\Omega)$, and $\mu', \nu' \in \mathcal{P}(\Omega')$ obtained by merging the last two elements of $\Omega$ into a single element. Then $\JP(\mu, \nu) \le \JP(\mu', \nu')$, with strict inequality if both $\mu, \nu$ have nonzero masses on those two elements.
\end{lemma}
\begin{proof}
By considering $\JP(\mu, \nu)$ as the probability that the argmin's of two lists of independent exponentials land on the same index $1 \le i_* \le n$, and using the fact that the minimum of two independent exponentials is an exponential with the sum of the rate parameters, we can couple the four argmin's arising from $\mu, \nu, \mu', \nu'$ and conclude by inspection.
\end{proof}
The lemma shows that any partition of $\Omega$ will lead to a $\JP$ that's greater than or equal to the original $\JP$. So the infimum is achieved with the most refined partition, namely $\Omega$ itself.

Finally we show that A* sampling applied to $\mu$ and $\nu$ simultaneously has a collision probability equal to $\JP(\mu, \nu)$ as defined above.
\begin{theorem}
Given two probability measures $\mu$ and $\nu$ on an arbitrary Polish space $\Omega$, both absolutely continuous with respect to a common third measure $\lambda$, it is possible to apply A* sampling with base distribution $\lambda$, either in-order or with a hierarchical partition of $\Omega$, to sample from $\mu$ and $\nu$ simultaneously. Further, the probability of the procedure terminating at the same point $p \in \Omega$ for both $\mu$ and $\nu$ is exactly $\JP(\mu, \nu)$. 
\end{theorem}
\begin{proof}
The first statement follows from the procedural definition of A* sampling described in \cite{astar}. For the second statement, since in-order A* is proven equivalent to hierarchical partition A* in \cite{astar}, we are free to choose any partition to our convenience. The natural choice is then the partition used in the definition of $\JP(\mu, \nu)$. 

More precisely, we know there is a finite partition $\cF \vdash \Omega$ such that $\JP(\mu, \nu) \leq \JP(\mu_\cF, \nu_\cF) \leq \JP(\mu, \nu) + \epsilon$, for any $\epsilon > 0$. On the other hand, for finite partition like $\cF$, the exponential variables attached to the representative of each part $Q \in \cF$ are jointly distributed as perfectly correlated exponentials with rate $\lambda(Q)$. Let $U, V$ be the two coupled A* processes restricted to $\cF$. Either one of them does not terminate, or they both terminate and collide conditionally with probability $\JP(\mu_\cF, \nu_\cF)$. In other words, letting $\p(T)$ be the probability that both terminate at $\cF$ level, and $AC(U, V; \cF)$ be the collision probability of $U, V$ restricted to $\cF$, then $ AC(U, V; \cF) = \p(T) \JP(\mu_\cF, \nu_\cF).$ Thus $\JP(\mu, \nu) \p(T) \leq AC(U, V; \cF) \leq \JP(\mu_\cF, \nu_\cF) \leq \JP(\mu, \nu) + \epsilon$.

So $AC(U, V; \cF)$ is squeezed between $\JP(\mu, \nu) \p(T)$ and $\JP(\mu, \nu) + \epsilon$. Since $\p(T) \to 1$ and $\epsilon \to 0$ under a refinement sequence $\cF$, we get in the limit
\begin{align*}
AC(U, V) := AC(U, V; \cF_\infty) = \JP(\mu, \nu).\quad\qedhere
\end{align*}
\end{proof}

\section{Utility of $\JP$ on Pairs of Web Documents}

\label{web-document-section}
To determine whether the difference between $\JP$ and $\JW$ matters in practice and whether achieving $\JP$ as a collision probability is a useful goal, we computed both for a large sample of pairs of unigram term vectors of web documents. From an index of \ifanon 6.6 billion \else the 6.6 billion highest PageRank \fi documents, we selected pairs using a sum of 2 $\sW$-MinHashes to perform importance sampling. We computed several similarity scores for 100 million pairs of normalized unigram term vectors, and weighted them by the inverse of their sampling probability to simulate an unbiased sample of all non-zero pairs.

The Jensen-Shannon divergence (JSD) defines the information loss that results from representing two distributions using a model that is an equal mixture of them, and as such is the ideal criterion to form information-preserving clusters of items of equal importance. $\text{JSD}(x,y) = \frac{1}{2}D_{\text{KL}}\left(x\Vert \frac{x+y}{2}\right) + \frac{1}{2}D_{\text{KL}}\left(y\Vert \frac{x+y}{2}\right)$. Like both $\JW$ and $\JP$ it is bounded, symmetric, and monotonic in a metric distance. Due to these properties and its popularity, we use it as a basis for comparison.

$\JP$ has a much tighter relationship with the Jensen-Shannon divergence than $\JW$ as shown in figure \ref{jsd-dist}. Tight bounds on JSD as a function of $\JW$ are given by $\JW$'s monotonic relationship with Total Variation, as described by \cite{sason2015tight}. Let $p$ be the total variation distance, and $d(p) = \frac{(1-p)}{2}\log_2(1 - p) + \frac{(1 + p)}{2}\log_2(1 + p)$. $d(p) \geq \text{JSD}(x,y) \geq p$. Substituting $\frac{1-\JW}{1+\JW} = p$ extends these to $\JW$. These same bounds apply to $\JP$ as well, but since $\JP\geq\JW$ it is always closer to the upper bound than $\JW$ and thus has a tighter relationship with what appears to be a much higher lower bound.

We have approached finding this lower bound by numerically solving an associated system of Euler-Lagrange differential equations, but small examples form good approximate bounds. On 2 element distributions, JSD has a direct relationship with $\JP$, and only $1\times 10^{-7}$ of the pairs fall below the resulting curve, $d(1-\JP)$. No pairs in our sample had JSD more than 0.0077 below it. In contrast, $\JW$ puts $7\times 10^{-3}$ of the pairs below this curve, with the farthest point 0.16 below. (Figure \ref{jsd-dist}.)

We also compare both $\JP$ and $\JW$ to the Jaccard index of the set of terms, and compute a kernel density estimate of the log of their ratios in figure \ref{jp-jw-dist}. $\JP$ is generally centered around the Jaccard index, while $\JW$ is consistently centered below, as predicted by their behavior on uniform distributions. This makes $\sP$-MinHash less disruptive as a drop-in replacement for an unweighted MinHash. Parameters of the system such as the number of hashes or the length of concatenated hashes are likely to continue to function well.

In the typical case of retrieval using a key-value store, performance is characterized by cheap ANDs and expensive ORs.\cite{datar2004locality} To reduce the collision probability we can sum multiple hashes to form keys, but to raise the collision probability, we must output multiple independent keys. This lets us apply an asymmetric sigmoid to the collision probabilities, $1 - (1 - p^a)^o$ with $o$ independent outputs of $a$ summed hashes each.  Assuming that the cost of looking up hashes dominates the cost of generating them, ANDs are essentially free, while CPU and storage cost are both linear in the number of ORs. Furthermore, as $a$ increases linearly, $o$ must increase exponentially to keep the inflection point of the sigmoid in the same place. For instance, if the sigmoid passes through $(0.5,0.5)$, then $o \approx \log(2)2^a$. This gives a significant performance advantage to algorithms with higher collision probabilities, and thus to $\JP$ over $\JW$. Lowering the probability is much cheaper than raising it.

The effect of this is demonstrated in Figure \ref{precrec}. Unsurprisingly from the tightness of the joint distribution, $\sP$-MinHash achieves better precision and recall retrieving low JSD documents for a given cost. More surprising is that it also achieves slightly better precision and recall on retrieving high $\JW$ documents when the cost is low, even though this is the task $\sW$-MinHashes are designed for. The reason for this can be seen from the upper bound, $\JP \leq 2\JW/(1+\JW)$. On items that achieve this bound, the collision probability when summing two hashes, $(2x/(1+x))^2$, is similar to $4x^2$ near 0 and similar to $x$ near 1. This in effect gives it the recall of 1 hash with the precision of 2 hashes on this subset of items, and thus a better precision/recall trade-off overall.

\section{Conclusion}

We've described a new generalization of the Jaccard index, and shown several qualities that motivate it as the natural extension to probability distributions. In particular, we proved that it is optimal on all distributions in the same sense that the Jaccard index is optimal on uniform distributions. We've demonstrated its utility by showing $\JP$'s similarity in practice to the Jensen-Shannon divergence, a popular clustering criterion. We've described two MinHashing algorithms that achieve this as their collision probability with equivalent running time to the state of the art on both sparse and dense data.

\bibliography{Minhash}

\end{document}